\documentclass[12pt]{article}
\input{my_styling.sty}

\usepackage{setspace}

\usepackage[margin=1in]{geometry}
\usepackage{microtype}
\usepackage{graphicx}
\usepackage{subfigure}
\usepackage{booktabs} 

\usepackage{hyperref}

\usepackage{amsmath}
\usepackage{amssymb}
\usepackage{mathtools}
\usepackage{amsthm}

\usepackage[capitalize,noabbrev]{cleveref}

\theoremstyle{plain}
\newtheorem{theorem}{Theorem}[section]
\newtheorem{proposition}[theorem]{Proposition}
\newtheorem{lemma}[theorem]{Lemma}
\newtheorem{corollary}[theorem]{Corollary}
\theoremstyle{definition}
\newtheorem{definition}[theorem]{Definition}

\theoremstyle{remark}
\newtheorem{remark}[theorem]{Remark}

\usepackage[
backend=biber,
style=alphabetic,
sorting=ynt
]{biblatex}
\DeclareUnicodeCharacter{2113}{}  

\title{Fast Sampling Based Sketches for Tensors\thanks{Research of William Swartworth and David P. Woodruff was supported by a Simons Investigator Award, a Google Faculty Award, and NSF Grant No. CCF-2335412.}}

\author{William Swartworth (CMU) \thanks{wswartwo@andrew.cmu.edu}
\and
David P. Woodruff (CMU) \thanks{dwoodruf@cs.cmu.edu}}

\date{}

\begin{document}
\maketitle

\begin{abstract}
We introduce a new approach for applying sampling-based sketches to two and three mode tensors.  We illustrate our technique to construct sketches for the classical problems of $\ell_0$ sampling and producing $\ell_1$ embeddings. In both settings we achieve sketches that can be applied to a rank one tensor in $(\mathbb{R}^d)^{\otimes q}$ (for $q=2,3$) in time scaling with $d$ rather than $d^2$ or $d^3$.  Our main idea is a particular sampling construction based on fast convolution which allows us to quickly compute sums over sufficiently random subsets of tensor entries.
\end{abstract}

\section{Introduction}

In the modern area of enormous data sets, space is often at a premium, and one would like algorithms that either avoid storing all available data, or that compress existing data.  A common and widely-applied strategy is \textit{sketching}.  Given a vector $x \in \R^n$ consisting of the relevant data, a (linear) sketch of $x$ is given by $Sx$ where $S$ is a linear map down to a dimension much smaller than $n$.  Typically the goal is to design $S$ so that some useful statistic of $x$ can be computed from the sketched vector $Sx$, even though most of the information about $x$ has been discarded.  Linear sketches are particularly useful in the context of streaming algorithms, since linear updates to $x$ can be translated to the sketch, simply by sketching the vector of updates, and adding it to the previous value of the sketch. Sketches have also found important applications in speeding up linear-algebraic (and related) computations \cite{woodruff2014sketching}.  Here, the idea is to first apply a sketch to reduce the dimensionality of the problem, while approximately preserving a quantity of interest (e.g., a regression solution). Then standard algorithms can be applied to the smaller problem resulting in improved runtimes.

A lot of work on sketching has focused on efficiently applying sketches to structured data.  For example, if the underlying data is sparse, one might hope for a sketch that can be applied in input-sparsity time.  A different type of structure that has been widely studied in this context is tensor structure. A line of work has been devoted to developing fast tensor sketches \cite{pham2013fast, ahle2020oblivious, ahle2019almost, meister2019tight}, which are sketches that can be applied quickly to low-rank tensors.  A rank-one tensor in $\R^n \otimes \R^n$ for example (i.e., a rank-one matrix) only requires $O(n)$ parameters to specify.  However, a naive sketch might require $O(n^2)$ time to apply if each entry of the tensor must be calculated to form the sketch.  The goal with tensor sketches is to do better then expanding into a vector form prior to sketching.

Much of the work based on sketching tensors has focused on $\ell_2$-norm settings.  The most studied example is Johnson-Lindenstrauss (JL) embeddings, which have been studied extensively for tensors.  For JL sketches, the ideal sketch is a dense Gaussian sketch.  However, such sketches are slow to apply and so the main interest is in approximating the properties of a Gaussian sketch with a sketch that can be applied faster.

We study a completely different type of ``sampling-based" sketch, which is embodied by the $\ell_0$ sampling problem.  The goal here is to observe a linear sketch of $x$, and then to (nearly) uniformly sample an element from $x$'s support.  A dense Gaussian sketch is completely useless as we would like a row of our sketch to single out an element of $\supp(x)$.  To allow for this, the standard idea is to take a sketch which performs sparse samples at various scales.  We ask: \textit{Are there sampling-based sketches which can be applied quickly to rank-one tensors?}  

For tensors with two and three modes (which cover many tensors of practical importance), we provide a new approach for constructing sampling-based sketches in the tensor setting.  To illustrate our approach, we focus on two fundamental problems: $\ell_0$ sampling, and $\ell_1$ embeddings. We recall the setup for these problems.

\paragraph{$\ell_0$ Sampling.} For the $\ell_0$ sampling problem one would like to construct a distribution over sketching matrices $\mathbf{S}$ so that observing $\mathbf{S} x$ allows us to return a uniformly random element from the support of $x$.  We allow constant distortion, and $\delta$ failure probability, so conditioned on not failing (which must occur with probability at least $1-\delta$), we should have that for all $i$ in the support of $x$, the sampling algorithm outputs $i$ with probability in $[c_1/|\supp(x)|, c_2/|\supp(X)|]$.

\paragraph{$\ell_1$ Embeddings.}  For the $\ell_1$ embedding problem, the goal is to construct a sketch $\mathbf{S}$, such that for any $x\in\R^n$, we have \[c_1 \norm{x}{1} \leq \norm{\mathbf{S}x}{1} \leq c_2 \norm{x}{1}\] for absolute constants $0 < c_1 < c_2.$

%



\subsection{Our Results}

Our main idea is constructing a way to subsample a tensor at a given sampling rate $p$, and then to sum over the resulting sampling.  Specifically we introduce a sampling primitive that we call a ``$p$-sample" which samples each element with probability approximately $p$, and does so in a nearly pairwise independent manner.  As we show, this will be sufficient for constructing both $\ell_0$-samplers and $\ell_1$ embeddings.  Constructing $p$-samples is straightforward for tensors which are given entrywise -- one can simply sample each entry independently with probability $p.$  Our key novelty is in showing that it is possible to construct $p$-samples for two and three mode tensors, which for rank-one tensors can be summed over in nearly linear time.  We discuss the ideas for constructing samples below.

Given our sampling primitive, we show how to construct $\ell_0$ samplers and $\ell_1$ embeddings.

For the $\ell_0$ sampling problem we show the following result:

\begin{theorem}
For $q = 2,3$ there there is a linear sketch of $X\in \R^{n^q}$  with sketching dimension  $m = O(\log\frac{1}{\delta}\log^2 n (\log\log n + \log\frac{1}{\delta}))$ space, and a sampling algorithm that succeeds with probability $1-\delta$, which, conditioned on succeeding, outputs an index $i$ in $\supp(X).$  Conditioned on succeeding, the probability that the algorithm outputs $i\in \supp(X)$ is in $[\frac{0.5}{|\supp(X)|}, \frac{2}{|\supp(X)|}].$  The algorithm also returns the value of $X_i.$  Moveover, the entries of the sketching matrix can be taken to be in $\{0,+1,-1\}.$

For $q=2$, our sketch can be applied to rank-one tensors in $O(m n)$ time, and when $q=3$ can be applied in $O(m n\log^2 n)$ time.
\end{theorem}

For $\ell_1$ embeddings, we show the following:

\begin{theorem}
There is an $O(1)$-distortion $\ell_1$ embedding sketch $\ellonesketch$ $\R^n \rightarrow \R^m$ with sketching dimension $m = O(\log^4 n + \log^2(1/\delta)\log n)$ that satisfies $\norm{\ellonesketch x}{1} \leq c \norm{x}{1}$ with constant probability, and satisfies $\norm{\ellonesketch x}{1} \geq c \norm{x}{1}$ with probability at least $1-\delta.$

Moreover our sketch can be applied to rank one tensors in $\R^{k}\otimes \R^{k}$ in $O(m k)$ time, and to rank one tensors in $\R^{k}\otimes \R^{k} \otimes \R^k$ in $O(m k\log^2 k)$ time.
\end{theorem}
Our main novelty here lies in the $p$-sample construction, which is applied in a similar way as for our $\ell_0$-sample, although the details are more complicated.  We therefore defer the proof of this result to the appendix.

\paragraph{Regression} We note that this error guarantee of ``no contraction" with high probability, and ``no dilation" with constant probability is know to be sufficient to construct algorithms for solving $\ell_1$ regression.  In particular, by setting $\delta = \exp(-d)$, a standard net argument given in \cite{clarkson2014sketching} for example, shows that ``no contraction" holds with high probability over a $d$-dimensional subspace.  As long as ``no dilation" holds for the solution vector, then it is well-known that this yields a dimension-reduction for $\ell_1$ regression.  Thus our $\ell_1$ sketch can be used to speed up sketching for an $\ell_1$ regression problem of the form $\min_x \norm{Ax - b}{1},$ where $A$ has $d$ columns, each of which have low-rank structure.


Finally, while we choose to illustrate our approach on the problems of $\ell_0$ sampling and $\ell_1$ embeddings, we believe that this technique could be applied more generally whenever there is a known sketch that is built from random subsampling and taking random linear combinations within the sample.



\subsection{Techniques}

The main idea behind both $\ell_0$-sampling and $\ell_1$-embeddings is to subsample coordinates at various scales in order to isolate some coordinates of $x.$ As a warmup, we begin by describing $\ell_0$ sampling. Then we describe how our techniques extend to constructing $\ell_1$ embeddings. We note that our basic approach to constructing $\ell_0$ samplers and $\ell_1$ embeddings are not novel.  However our ability to quickly sketch rank-one tensors is.

\paragraph{$\ell_0$ Sampling}
In the general form of $\ell_0$ sampling, we are given a vector (or in our case a tensor) $X$, and the goal is to design a sketch that allows us sample an entry of $X$ nearly uniformly from the support of $X.$  The idea is to take a random sample $S$ of $X$'s entries and then to store a random linear combination of the values in $S$ as an entry of the sketch.  If $S$ intersects $\supp(X)$ in a single element then this will allow us to recover a value in $\supp(X)$.  If $\supp(X)$ has size $k$, then we would like $S$ to sample each value with probability roughly $1/k$ in order to have a good probability of isolating a single element.  Since $k$ is unknown to us initially, the idea is to perform the sampling procedure at each of $\log n$ sampling levels $1,\frac{1}{2}, \frac{1}{4}, \ldots$.

The challenge for us is to design samples $S$ that allow for fast sketching on rank-one tensors, e.g., tensors of the form $x\otimes y$.  As we later show, the step of computing a random linear combination of values in $S$ can be achieved simply by randomly flipping the signs along each mode.  The harder part is designing $S$ so that we can sum the resulting tensor over $S$ in roughly $O(n)$ time.  In other words we need $S$ to be such that we can quickly compute
$
\sum_{(i,j) \in S} x_i y_j.
$
If $S$ just samples each entry of $X$ i.i.d. with probability $p$, then it is not clear how to do better than computing the sum term-by-term, which would require $\Omega(pn^2)$ time in our example.

The next most natural thing to try is to take $S = S_1 \times S_2$ to be a random rectangle, since rectangles allow the above summation to be computed in $O(n)$ time.  To achieve sampling probability $p$ one could take a rectangle of dimension $\sqrt{p} d \times \sqrt{p} d,$  where the subsets of indices $S_1$ and $S_2$ corresponding to each dimension are chosen randomly.  Thus, construction succeeds in sampling each entry of $X$ with probability $p$; however this is not sufficient to isolate a single entry of $\supp(X)$ with good probability.  For example $\supp(X)$ could consist of a single row of $X$.  If we set $p=1/d$ then $S$ has a good probability of sampling some entry from $\supp(X)$. However, when this occurs, $S$ nearly always contains more than one such element.  Indeed, $S_2$ would typically contain around $\sqrt{d}$ elements, and in this situation it is not possible to sample a single entry from a row of $X.$  One could try to fix this by randomizing the sampling scale on each side of the rectangle for a fixed sampling probability. In the case of $\ell_0$ sampling this would at least result in worse logarithmic factors in the sketching dimension than what we achieve.  In the case of $\ell_1$ embeddings, it is unclear how to obtain a constant-factor embedding with this approach.  For instance, in the two-mode case, a given sampling probability $p$ can be realized by $O(\log d)$ choices of scales for the two rectangle dimensions.  If a given level of $X$ has size $1/p$ and is distributed uniformly throughout $X$, then all rectangles with sampling probability $p$ will pick out single elements of that level.  On the other hand, if the level is distributed along a single $1$-dimensional fiber, then it might be the case that only one of the rectangles is good for that level.  This suggests that the distortion would likely need to be $\Omega(\log d)$ on some inputs.  It does not seem clear how one could arrange arrange for an $O(1)$ distortion $\ell_1$ embedding that handles these situations simultaneously.

Surprisingly, at least for $2$ and $3$ mode tensors, it is possible to construct samples $S$ that have better sampling properties than rectangles, but which can still be summed over in nearly linear time.  For brevity, we refer to the sampling that we need as a $p$-sample, when it samples each index with probability (approximately) $p.$ The idea is to design $S$ in such a way that we can employ algorithms for fast convolution using the Fast Fourier Transform. For example, in the $2$-mode case, one can compute
\[
\sum_{i + j \in [0,T-1]} x_i y_j 
\]
for a constant $T$ (chosen to give the appropriate sampling probability) in $O(n\log n)$ time, by first calculating the convolution $x * y$ and then summing over the indices in $[0,T-1].$  This only gives a sum over a fixed set, but by randomly permuting the indices along each mode, and choosing $T$ appropriately, this allows us to achieve sampling probabilities down to $p = 1/n.$  Smaller sampling probabilities result in samples of size $O(n)$, and so the sum can just be calculated explicitly.
(As it turns out, in the $2$-mode case, the runtime can be improved to $O(n)$ by a simple optimization.)

The $3$-mode case is somewhat more complicated.  A similar convolution trick involving a convolution of three vectors allows us to construct $p$-samples with fast summation down to $1/n.$ And sampling probabilities below $1/n^2$ allow for direct computation.  What about sampling probabilities between $1/n^2$ and $1/n$? Here, we show how to compute sums of the form
\[
\sum_{i+j+k=0, j-i\in[0,T-1]} x_i y_j z_k
\]
in $O(n\log^2 T)$ time.  The rough approach is to show that we can reduce computing this quantity to $n/T$ instances of multiplication by the top-left corner ($\{i,j : i \geq j\}$) of a roughly $T\times T$ Toeplitz matrix.   Each top corner matrix can be decomposed into a sum of $T\log T$ Toeplitz matrices (up to some zero-padding) each of which has an $O(T\log T)$ multiplication time.

\paragraph{$\ell_1$ embeddings}  For $\ell_1$ embeddings our approach roughly follows the sketch introduced by \cite{verbin2012rademacher} for producing constant distortion $\ell_1$ embeddings.  The main idea is to think of a vector $x\in \R^n$ as decomposed into approximate level sets, with the size of entries in each level set decreasing exponentially.  For instance, when $\norm{x}{1}=1$, the level sets $L_i$ could be taken to contain entries in $[q^{-(i+1)}, q^{-i}]$ for some $q < 1.$  For each level set, our sketch has an associated level with sampling probability larger than $q$, which is designed to capture most of the mass of $L_i$. Since we would like our sampling level to capture the mass of $L_i$ with high probability, we need to oversample the entries of $L_i$ substantially.  The usual approach here is to choose a fairly large $q$, and then hash the sampled values into separate buckets so as to minimize cancellations.  This is typically done with CountSketch; however it is not clear how to efficiently compose a CountSketch with our fast sample constructions.  Therefore we take a slightly different approach.  Instead of hashing each sampling level into $T$ buckets, we simply take $T$ independent samples, each with sampling probability $q/T.$. The analysis is quite similar to the analysis using CountSketch, but we retain the the ability to quickly sketch rank-one tensors.  However, we do pay a price -- our sketching time scales linearly with the size of the sketch. In contrast, for vectors, CountSketch constructions allow for sketching in time that scales as $O(n\cdot({\text{number of sampling levels}}))$.  However these sketches are not efficient for rank-one tensors, since $n$ would be replaced by $n^2$ or $n^3.$

To avoid too much dilation, the rough idea is that each sampling level picks up roughly the right amount of mass from its corresponding $L_i$ with good probability.  Also, large levels have a small numbers of elements and so they are unlikely to be picked up by a smaller sampling level.  The main issue is in bounding the contribution of the $L_i$'s containing many small values.  To avoid, this \cite{verbin2012rademacher} introduced the idea of hashing each sampling level into buckets with random signs in order to induce cancellation among the contributions from such $L_i$'s.  We employ the same technique here.  However in order to make applying the signs efficient, our approach is to first apply random signs along the modes of the tensor, and then to sum over appropriate $p$-samples in order to compute a ``bucket".  For a constant number of modes, this turns out to be enough randomness to induce cancellation.

This outline for constructing $\ell_1$ embeddings has appeared in the literature several times.  Our main novelty lies in constructing samples that admit fast linear combinations for rank-one tensors.

Finally, one might wonder if our techniques are really necessary. We could hope to design fast sketches using Kronecker-structured-sketches of the form $S_1 \otimes S_2$, which are easy to apply to a rank-one tensor $x\otimes y$ -- one simply computes $S_1 x \otimes S_2 y.$  It does not appear that a Kronecker-structured sketch can match our bounds.  Indeed by a lower bound given in \cite{li2021exponentially} , an $O(1)$ $\ell_1$ embedding for a single vector $x\in\R^n$ requires at least $\poly(n)$ sketching dimension for Kronecker-structured sketches.  On the other hand, our sketch can still be applied to rank-one tensors in near-linear time, but requires only $\poly\log n$ space to embed $x.$   One could also ask whether Kahtri-Rao sketches (i.e., a sketch with each row of the sketch a rank one tensor) could be applied to match our sketching dimension.  While we do not provide a lower bound against Kahtri-Rao measurements, we also do not see how to match our bounds for either $\ell_0$ sampling or $\ell_1$ embeddings.  We leave the question of whether such a lower bound exists as an interesting open question.

For larger $q$ we also note that our $\ell_1$ embedding procedure could be applied to triples of modes at a time, giving a tree construction similar to \cite{ahle2020oblivious}.  Unfortunately however, the distortion grows like $2^{\log_3 q}$, and the success probability becomes $c^{-q}$ by a union bound.  Whether these parameters can be improved while allowing for fast sketching is an interesting problem.

\subsection{Additional Open Questions}  An interesting open question is whether a similar set of ideas can be applied to tensors with a larger number of modes. For four modes, the most natural approach here would be to attempt to develop a fast sum that can sum over a (subset of) a codimension $2$ subspaces of the tensor entries.  For example, we might wish to calculate sums of the form
\[
\sum_{i,j,k,\ell \in Q} w_i x_j y_k z_{\ell}
\]
where $Q$ is a set of entries satisfying two linear constraints. Unfortunately, we are not aware of a way to calculate such a sum faster than $O(n^2)$ time. It would be interesting to either give a fast summation algorithm, or to find a new technique that gets around this issue.

\subsection{Related Work}
A substantial literature has been devoted to obtaining fast sketches for tensors.   Tensor sketching was initiated in \cite{pham2013fast} where it was shown that CountSketch can be quickly applied to rank-one tensors using the Fast Fourier Transform. More recently, $\ell_2$ embeddings were constructed by \cite{ahle2019almost} and improved in \cite{ahle2020oblivious} with applications for sketching polynomial kernels.

In earlier work \cite{indyk2008declaring} gives a Kronecker-structured for the $\ell_1$ norm in the context of independence testing, using a variation on the previously-known Cauchy sketch \cite{indyk2006stable}. We note that this sketch requires taking a median in order to recover the $\ell_1$ norm, and thus does not give an $\ell_1$ embedding, which may be more suitable for optimization applications.

\cite{indyk2006stable} studies the problem of $\ell_1$ estimation using Cauchy sketches. 
 Later, \cite{verbin2012rademacher} gave a construction of an oblivious $\ell_1$ embedding.  Our general approach for constructing $\ell_1$ embeddings is largely based on theirs.  This approach is also generalized in \cite{clarkson2014sketching} to M-estimators, and in particular is applied to construct $\ell_1$ (and more general) subspace embeddings.  These bounds for $\ell_1$ embeddings were recently improved in \cite{munteanu2021oblivious}.

\cite{li2021exponentially} expanded on the  work and considered independence testing for higher order tensors.  They also give a $\poly(d)$ lower bound for constructing $\ell_1$ embeddings for a single vector in $\R^d$ using Kronecker-structured measurements.

\subsection{Notation and Preliminaries}

We use the notation $[N]$ to refer to the set $\{1,\ldots, N\}.$  For two vectors $x$ and $y$, the notation $x\otimes y$ refers to the Kronecker product of $x$ and $y.$  We use the notation $x * y$ for the circular convolution of $x$ and $y$. That is
\[
(x*y)_k = \sum_{i+j = k} x_i y_j,
\]
where the sum $i+j$ is interpreted mod $N.$  In such situations it is sometimes convenient to index vectors starting at $0$ so we will do this occasionally.  However unless otherwise stated, we index starting at $1.$

We typically identify a $q$-mode tensor with a vector in $\R^{n^q}$, assuming that the dimension along each tensor mode is $n.$ We will typically make this assumption for the sake of convenience.  We say that a tensor in $\R^{n^q}$ has rank one if it can be expressed as $x_1 \otimes \ldots \otimes x_q$ for some vectors $x_1, \ldots, x_q.$

We will use $c$ throughout to refer to an absolute constant, which might be different between uses (even within an equation).

The notation $\tilde{O}(f)$ means $O(f\log^c f)$ for some constant $c.$

\paragraph{$\ell_0$ sampling.}
$\ell_0$ sampling has received extensive attention.  \cite{cormode2014unifying} surveys the standard recipe for constructing $\ell_0$-samplers, which we roughly follow here.

\cite{woodruff2018distributed} considers solving the $\ell_0$-sampling problem in a distributed setting for a product of two matrices which are held on different servers, however this is different from our setting, as their communication scheme is not a linear sketch.


\section{$p$-sample constructions}

We first define our main sampling primitive, which we call a $p$-sample.  This can be viewed as a slightly weaker version of a pairwise independent sample.
\begin{definition}
Let $T$ be an arbitrary finite set. We say that a random subset of $S$ of $T$ is a $p$-sample for $T$ if
\begin{enumerate}
\item For all entries $i\in T$, 
\[
p/2 \leq \Pr(i \in S) \leq p.
\]
\item For all $j\neq i$ in $T$, 
\[
\Pr(j \in S | i \in S) \leq 2p.
\]
\end{enumerate}
\end{definition}

\begin{definition}
Fix an $n$ and $m$ and suppose that $S$ is a subset of $[n]^{m} = [n]\times \ldots \times [n]$.  Let $x^{(1)}, \ldots, x^{(m)} \in \R^n$ be arbitrary vectors.  We say that $S$ admits fast summation if there is an algorithm which computes $\sum_{(i_1, \ldots, i_m)\in S} x_{i_1}^{(1)}x_{i_2}^{(2)}\ldots x_{i_m}^{(m)}$ in $\tilde{O}(mn)$ time.
\end{definition}

To perform one of our constructions we will need a particular type of Toeplitz-like fast matrix multiplication.
\begin{lemma}
\label{lem:fast_toeplitz_corner}
Let $A\in \R^{n\times n}$ be a matrix where $A_{i, j} = A_{i+1, j+2}$ for all $0 \leq i \leq n-1$, $0\leq j \leq n-2.$  Let $B$ be defined by $B_{i,j} = A_{i,j}$ if $j \geq i$ and $B_{i,j} = 0$ otherwise.

Given $v\in \R^n$ the matrix product $A v$ can be computed in $O(n\log n)$ time, and the product $B v$ can be computed in $O(n\log^2 n)$ time. 
\end{lemma}

\begin{proof}
To see that $A$ admits fast multiplication, note that the rows of $A$ coincide with the even rows of a Toeplitz matrix of size $2n \times 2n.$ Since multiplication by a Toeplitz matrix can be carried out in $O(n\log n)$ time \cite{kailath1999fast}, the same holds for $A.$

To get a fast algorithm for $B$ we decompose it into a sum of matrices, each with the same structure of $A$ up to some zero-padding.  We call these matrices $B_1, B_2$ and $B_3.$  Take $B_1$ to be $B$ but with all indices outside of $\{(i,j) : 1\leq i \leq \lceil n/2 \rceil, \lceil n/2 \rceil \leq j \leq n\}$ replaced with $0.$ Visually, the support of $A$ is a right triangle and the support of $B_1$ corresponds to the largest square inscribed in $A.$  Let $B_2$ and $B_3$ correspond to the two triangles that are left after removing the square. That is, $B_2$ has support contained in $\{(i,j): i > \lfloor\frac{n}{2}\rfloor\}$ and $B_3$ has support contained in $\{(i,j): j > \lfloor\frac{n}{2}\rfloor\}.$

By construction, $B = B_1 + B_2 + B_3.$  Then $B_1$ has the structure of $A$ from the first part of the lemma, so admits $O(n\log n)$ time multiplication.  After removing the zero padding, both $B_2$ and $B_3$ have the same structure as $B$ but with half the dimension.  Thus we have a runtime recurrence of the form $T(n) = 2 T(\lfloor n/2\rfloor) + O(n\log n)$, which gives an overall runtime of $O(n\log^2 n).$

\end{proof}

\begin{theorem}
For $m = 2,3$, and for all $n$ and $p$, there exists a $p$-sample for $[n]^m$ that admits fast summation.  Additionally
\begin{enumerate}
    \item When $m=2$ the summation runs in $O(n)$ time
    \item When $m=3$ the summation runs in $O(n\log^2 n)$ time
\end{enumerate}
\end{theorem}

\begin{proof}

We first show how to construct sets of indices that can be summed over quickly.  We will then show how to use these sets to construct $p$-samples with fast summation.

\textbf{m=2 Constructions.} Let $T$ be an arbitrary positive integer, and consider the set
\[
A_T = \{(i,j) \in [n]\times [n] : i + j \in [T]\} \subseteq (\mathbb{Z}/n)^2.
\]

(Note that we are treating all indices as values in $\mathbb{Z}/n.$) We will first show that $A_T$ admits $O(n)$ summation for all $T.$  We would like to compute a sum of the form
\[
\sum_{i+j \in [T]} x_i y_j = \sum_i \sum_{j:j-i \in T} x_i y_j
=\sum_i x_i \sum_{j \in [1+i, T+i]} y_j.
\]
Let $S(i)$ denote the value of the inner sum for a fixed $i$. Note that $S(i+1) = S(i) + y_{T+j+1} - y_{i+1}.$  Then $S(0)$ can be computed in $O(n)$ time and each of $S(1), S(2),\ldots, S(n-1)$ can be computed in turn, each in $O(1)$ time using the recurrence.  This gives an $O(n)$ algorithm for computing the original sum.

\textbf{m=3 Constructions.}  Let $T$ be arbitrary and set $B_T = \{(i,j,k) : i+j+k\in [0,T-1]\}$.consider a sum of the form
\[
\sum_{(i,j,k)\in B_T} x_i y_j z_k = \sum_{t\in [0,T-1]}\sum_{i+j+k = t} x_i y_j z_k.
\]
Each of the inner sums occurs as an entry in the circular convolution $x * y * z$.  A circular convolution can be computed in $O(n\log n)$ time using the Fast Fourier Transform, and so $B_T$ can be computed in $O(n\log n)$ time.

For $m=3$, we also need a sparser construction.  For this, define $C_T = \{(i,j,k): i+j+k = 0, j-i \in [0,T-1]\}.$ We are interested in the sum
\begin{align*}
\sum_{(i,j,k)\in C_T} x_iy_jz_k 
&=  \sum_{i+j+k = 0, j-i \in [T]} x_i y_j z_k\\
&= \sum_{k}z_{-k}\sum_{i+j=k, j-i\in [T]} x_i y_j.
\end{align*}
As the algorithm for fast multiplication is slightly more involved we defer to Lemma~\ref{lem:fast_CT_summation}.

\paragraph{Constructing $p$-samples.}  In each case we will begin by applying a random function $P_i : [n]\rightarrow [n]$ independently along each mode $i$.  Then we will take the indices that land in either $A_T$, $B_T$, or $C_T.$

For $A_T$ we take the set of indices 
\[
A_T' = \{(i,j): (P_1(i), P_2(j)) \in A_T \}.
\]

The probability that $(i,j) \in A_T$ is $T/n$ since $(i,j)$ is uniformly random.  For $(i,j)\neq (k,\ell)$ the values of $P_1(i) + P_2(j)$ and $P_1(k) + P_2(\ell)$ are independent. To see this, suppose without loss of generality that $j \neq \ell.$ Then conditioned on $P_1(i),P_2(j)$ and $P_1(k)$, we have that $P_2(\ell)$ is uniform over $[n].$  Therefore the events $(P_1(i), P_2(j)) \in A_T$ and $(P_1(k), P_2(\ell)) \in A_T$ are independent, and so 
\[
\Pr((k,\ell)\in A_T' | (i,j)\in A_T') = T/n.
\]

For $B_T$, precisely the same construction and argument applies.

$C_T$ requires a bit more work.  This time we choose $P_1, P_2, P_3$ to be random permutations.  Define $P$ by $P((i,j,k)) = (P_1(i), P_2(j), P_3(k)).$ As before, for any $(i,j,k)$, $(P_1(i),P_2(j), P_3(k))$ is uniformly random, so $(P_1(i),P_2(j), P_3(k)) \in C_T$ with $T/n^2$ probability.  Now consider two pairs $u_1 = (i_1, j_1, k_1)$ and $u_2 = (i_2, j_2, k_2)$ with $u_1 \neq u_2.$  Suppose that $u_1$ and $u_2$ differ in only a single coordinate.  Then since $P_1, P_2, P_3$ are permutations, 
\[
\Pr( P(u_1) \in C_T | P(U_2) \in C_T ) = 0,
\] 
since $C_T$ intersects each single-mode fiber in at most one coordinate.

Next we consider the case where $u_1$ and $u_2$ differ in precisely two coordinates.  We start the case where $u_1$ and $u_2$ agree in the third coordinate.  Without loss of generality, we assume that $u_1=(0,0,0)$ and $u_2 = (1,1,0).$ Then $\Pr(P(u_1) \in C_T\aand P(u_2)\in C_T)$ is the probability that the following events occur:
\begin{enumerate}
\item $P_1(0) + P_2(0) = -P_3(0)$
\item $P_2(0) - P_1(0) \in [0,T-1]$

\item $P_1(1) + P_2(1) = -P_3(0)$
\item $P_2(1) - P_1(1) \in [0,T-1]$
\end{enumerate}
For any fixed value of $P_3(0)$, there are exactly $T$ pairs $(P_1(0), P_2(0))$ that satisfy the first two equations, and similarly there are $T$ pairs $(P_1(1), P_2(1))$ that satisfy the last two equations. There are $T(T-1)$ ways to choose the two pairs, so the probability that $P_1$ and $P_2$ give two such pairs is $T(T-1)(\frac{1}{n}\frac{1}{n-1}\frac{1}{n}\frac{1}{n-1}).$ This implies that 
\[
\Pr(u_2 \in C_T | u_1 \in C_T) = \frac{T-1}{(n-1)^2}.
\]

A similar calculation gives the same probability when $u_1$ and $u_2$ agree on precisely the first or second coordinates.

Finally consider the case where $u_1$ and $u_2$ differ on all three coordinates.  Then $P_1(u_1)$ is uniform over all possible triples, and conditioned on $P_1(u_1)$, $P_2(u_2)$ is uniform over all triples that differ from $P_1(u_1)$ in all coordinates. There are $(n-1)^3$ such triples and at most $Tn$ of these are in $C_T.$  So 
\[
\Pr(u_2 \in C_T | u_1 \in C_T) \leq \frac{Tn}{(n-1)^3}.
\]

For two modes, we have constructed $p$-samples for $p=n/T$. This is sufficient to give a $p$-sample for all sampling probabilities down to $1/n.$  For sampling probabilities smaller than $1/n$, the size of the sample will be $O(n)$ with high probability, and so it admits fast summation by direct computation.

For three modes, our two constructions give $p$-samples down to sampling probability $1/n^2.$  For smaller $p$, we can again compute the desired sum explicitly in $O(n)$ time.
\end{proof}

We now verify that $C_T$ indeed admits fast summation.
\begin{lemma}
\label{lem:fast_CT_summation}
Let $T$ be a positive integer, and let $x,y,z\in \R^n$.  Let 
\[C_T = \{(i,j,k) : i + j + k = 0\,\, \text{and}\,\, i - j \in \{0,\ldots, T-1\},\]
where the arithmetic operations are treated mod $T$.  There is an algorithm that computes
\[
\sum_{(i,j,k)\in C_T} x_i y_j z_k
\]
in $O(n\log^2 T)$ time. For convenience, we zero-index into each vector. 
\end{lemma}
\begin{proof}
We first rewrite the sum as
\begin{align*}
\sum_{(i,j,k)\in C_T} x_i y_j z_k
&= \sum_{k} z_{-k} \sum_{(i,j): i+j = k, i-j\in \{0,\ldots, T-1\}} x_i y_j\\
&= \sum_{k} z_{-k} \sum_{i:2i \in \{k,\ldots, k+T-1\}} x_i y_{k-i}.
\end{align*}
Note that $\{i \in \mathbb{Z}/T : 2i \in [k, k+T-1]\}$ is a union of two intervals $I_k$ and $J_k$ in $\mathbb{Z}/T$ which are given by 
\begin{align*}
I_k &= \left[\left\lceil \frac{k}{2} \right\rceil, \left\lfloor \frac{k+T-1}{2} \right\rfloor\right],\,\,\\
J_k &= \left[\left\lceil \frac{k+T}{2} \right\rceil, \left\lfloor \frac{k+2T-1}{2} \right\rfloor\right].
\end{align*}
We now split the inner sum over $I_k$ and $J_k.$  We also split the outer sum over even and odd values of $k$ so that the intervals shift by one with each term.  This gives four sums to compute, each of which is similar.  For the first sum, we wish to compute
\begin{align*}
\sum_{k': 0\leq 2k' < T} z_{-2k'} &\sum_{i \in I_{2k'}} x_i y_{k-i}\\
&= \sum_{k': 0\leq 2k' < T} z_{-2k'} \sum_{i \in k' + I_{0}} x_i y_{2k'-i}.
\end{align*}
We will show how to compute all of the inner sums quickly.  Define the ``stride-$2$" circulant matrix $Y\in \R^{n\times n}$ by $Y_{kj} = y_{2k - j}.$  Define $Y'_{kj}$ by $Y'_{kj} = Y_{kj}$ for $j \in k + I_0$ and $0$ otherwise.  The sums we wish to evaluate are precisely the entries of $Y'x$ by construction. 

Now we show how to evaluate $Y' x.$ By increasing the dimension of $Y'$ and $x$ by at most $T$ (and continuing the circulant pattern), we may assume without loss of generality that $n$ is a multiple of $|I_0|.$

The support of $Y'$ contains a ``stripe" of width $\abs{I_0}$ running parallel to the diagonal.  This stripe decomposes into a union of $2n/T$ right triangles with legs of length $|I_0|,$ and with $n/T$ in each of two orientations.  More formally, by permuting rows, we may assume that $I_0 = [0, a)$ where $a=|I_0|.$  Then the triangles described are the restrictions of $Y'$ to sets of the form $[ra, (r+1)a] \times [sa, (s+1)a].$ Each such triangular submatrix is of the form considered in Lemma~\ref{lem:fast_toeplitz_corner}, so admits $O(|I_0|\log^2 |I_0|) = O(T\log^2 T)$ matrix-vector multiplication.  Since there are $2n/T$ such submatrices of $Y'$, matrix vector multiplication with $Y'$ runs in $O(n\log^2 T)$ time.  This allows the sum in question to be computed in the same time.

Each of the other three sums can be evaluated similarly and gives the same runtime.  So the runtime stated in the lemma follows.
\end{proof}

We also collect a couple of basic facts that follow from the definition of a $p$-sample.
\begin{proposition}
\label{prop:basic_psample_facts}
Let $S$ be a $p$-sample for $x$, and let $L$ be a subset of indices.  
\begin{enumerate}
    \item For all $i$ in $L$, we have $S\cap L = \{i\}$ with probability at least $\frac{1}{2}p(1-2p|L|).$

     \item With probability at least $\frac{1}{2}|L|p(1 - 2p|L|)$, $|S\cap L| = 1.$


\end{enumerate}
\end{proposition}
\begin{proof}
Consider a fixed index $i$ in $L.$  With probability at least $p/2$ we have $i\in S.$  Conditioned on $i\in S$, we have $\Pr(j \in S) \leq 2p$ for all $j\neq i.$  Taking a union bound over all $j\in L$ gives $\Pr(S\cap L = \{i\}) \geq p/2 (1 - 2p|L|).$  The events $S\cap L = \{i\}$ are mutually exclusive as $i$ ranges over $L$, so the proposition follows.
\end{proof}

\section{Constructing an $\ell_0$ sampler.}

Before giving the $\ell_0$-sampler we give two basic constructions for $1$-sparse recovery and singleton testing using Kahtri-Rao structured measurements. We want Kahtri-Rao measurements here, so that we can apply our random sign flips modewise.
\begin{lemma}
\label{lem:singleton_tester}
Let $X\in \R^{n^q}.$  Suppose that $\sigma^{(i)}_j$ for $i \in [N] $ and $j\in [q]$ are random sign vectors. For $N \geq c\log n \log\frac{1}{\delta}$, the measurements $\inner{\sigma^{(i)}_1 \otimes \ldots \otimes \sigma^{(i)}_q}{X}$ are sufficient to
\begin{enumerate}
\item recover $X$ with $1-\delta$ probability if $X$ is $1$-sparse\footnote{$X$ being $1$-sparse means that $X$ has at most one nonzero entry}.
\item to distinguish between $X$ being $1$-sparse and $X$ having at least two nonzero entries with at least $1-\delta$ probability.
\end{enumerate}
\end{lemma}
\begin{proof}
First note that if $X$ is $1$-sparse then all measurements will be equal in absolute value. 

Now, suppose that $X$ is $1$-sparse with $k$ in $\supp(X)$.  Say that the ``sign pattern" for $k$ is the sequence of signs of $\inner{\sigma^{(i)}_1 \otimes \ldots \otimes \sigma^{(i)}_q}{X}$ for $i\in [N].$  We can recover $k$ as long as its sign pattern is unique among indices, up to inverting all the signs (since an entry of $X$ could be either positive or negative).  The entries of each measurement are pairwise independent, so the probability that $k_1$ and $k_2$ have the same sign on a given measurement is $1/2.$ The probability that they have the same sign on each of $N$ measurements is $1/2^N$, so the probability that some other index has the same sign pattern as $k$, up to inverting signs, is at most $2 * n^q/2^N$ which is at most $\delta$ for $N \geq \log\frac{2 n^q}{\delta}.$  So if $X$ is $1$-sparse, then we recover $X$ with $1-\delta$ probability. 


To test if $X$ is $1$-sparse we first run the $1$-sparse recovery algorithm to obtain a $1$-sparse candidate $Y$ for $X$ (or if the recovery fails, then we immediately return False).  Now we would like to test if $X-Y$ is $0.$  To do this, we simply take additional sign measurements in the form of those in the lemma, and test if they all give $0.$ Clearly if $X-Y=0$ then this will occur.  Now consider the case where $X - Y \neq 0.$  Consider the natural compression $\R^{\otimes q} \rightarrow \R^{\otimes (q-1)}$ obtained by applying $\sigma_q$ along the $q$-th mode.  If $X-Y$ is non-zero then it has a 1-dimensional fiber $v$ along mode $q$ which is non-zero.  It is clear that $\inner{v}{\sigma_q}$ is non-zero with probability at least $1/2.$ Therefor applying this inductively shows that each measurement $\inner{\sigma^{(i)}_1 \otimes \ldots \otimes \sigma^{(i)}_q}{X-Y}$ is non-zero with probability at least $1/2^q.$  So $O(2^q \log\frac{1}{\delta})$ additional measurements suffice.
\end{proof}
\begin{remark}
\label{rmk:intialize_once}
The above algorithm shows that there exists a collection of sign measurements that allows for $1$-sparse recovery with probability $1.$ Thus such a collection of signs only needs to be sampled once, and it will be valid with probability $1-\delta.$  For the equality testing step, the randomness is essential.  However, conditioned on having a no-failure $1$-sparse recovery algorithm, we note that it yields no false negatives.  That is, if $X$ is $1$-sparse, then the tester will always accept.
\end{remark}

Our next lemma for $\ell_0$ sampling shows that if we succeed in isolating a single element of the support of $X$ at an $O(1/|\supp(X)|)$ sampling level, then that element is close to uniform from the support.
\begin{lemma}
\label{lem:near_uniform_lemma}
Let $S$ be a $p$-sample, where $p \leq \frac{1}{4|\supp(X)|}$.  Then for all $i \in \supp(X),$
\[
\Pr(i \in S \big| |S\cap \supp(X)| = 1) \in \left[\frac{0.25}{|\supp(x)|}, \frac{4}{|\supp(X)|} \right].
\]
\end{lemma}
\begin{proof}
See the supplementary.
\end{proof}

We are now ready to give our construction of an $\ell_0$ sampler.

\begin{theorem}
For $q = 2,3$ there there is a linear sketch of $X\in \R^{n^q}$  with sketching dimension  $m = O(\log\frac{1}{\delta}\log^2 n (\log\log n + \log\frac{1}{\delta}))$ space, and a sampling algorithm that succeeds with probability $1-\delta$, which, conditioned on succeeding, outputs an index $i$ in $\supp(X).$  Conditioned on succeeding, the probability that the algorithm outputs $i\in \supp(X)$ is in $[\frac{0.5}{|\supp(X)|}, \frac{2}{|\supp(X)|}].$  The algorithm also returns the value of $X_i.$  Moveover, the entries of the sketching matrix can be taken to be in $\{0,+1,-1\}.$

For $q=2$, our sketch can be applied to rank-one tensors in $O(m n)$ time, and when $q=3$ can be applied in $O(m n\log^2 n)$ time.
\end{theorem}
\begin{proof}
Our sketch consists of $O(\log n^q)$ $p$-samples with sampling probabilities $p=1,\frac{1}{2},\frac{1}{4}\ldots, \frac{1}{n^q}$.  Additionally, we take $k$ independent such samples for each $p$ which we denote by $S^{(p)}_1, \ldots , S^{(p)}_k.$

There is some $p$ in $[\frac{1}{8|\supp(X)|}, \frac{1}{4|\supp(X)|}.$  For this value of $p$ and fixed $i$, we have that $\abs{S^{(p)}_i \cap \supp(X)} = 1$ with probability at least $1/8$ by Proposition~\ref{prop:basic_psample_facts}. Thus, the probability that this holds for some $i$ is at least $1 - (7/8)^k,$ which is at least $1-\delta$ for $k\geq \log\frac{1}{\delta}.$

Our $\ell_0$ sampling algorithm is to to iterate through the samples in increasing order of $p$, and within each level to iterate through the $S^{(p)}_i$'s in increasing order of $i.$  For each $(p,i)$ we run the singleton tester from Lemma~\ref{lem:singleton_tester}.  If the singleton tester accepts, then we use the $1$-sparse recovery scheme to return the appropriate index and value.  As we have seen, with probability at least $1-\delta$, the algorithm will return for some $p \leq \frac{1}{4|\supp(X)|}.$ We also use the optimization of Remark~\ref{rmk:intialize_once}, so that we can assume the $1$-sparse tester never yields false negatives.

As long as this occurs, the algorithm will be successful by the previous Lemma~\ref{lem:near_uniform_lemma} unless the first $S_i^{(p)}$ for which the singleton accepts is not $1$-sparse.  Call this event $E.$  If $E$ occurs, then the singleton tester must fail on one of $O(k \log n)$ trials, so we would like to make its failure probability at most $O(\delta/(k\log n).$ By Lemma~\ref{lem:singleton_tester} this can be accomplished using a singleton-testing sketch of size $O(\log n \log(k\log n/\delta))$ which is $O(\log n\log\log n + \log n\log\frac{1}{\delta})$ when $k = O(\log \frac{1}{\delta}).$

Now, we have $k$ singleton sketches for each of $\log n$ sampling probabilities $p$, resulting in a total sketching dimension of $O(k\log^2 n (\log\log n + \log\frac{1}{\delta})) = O(\log\frac{1}{\delta}\log^2 n (\log\log n + \log\frac{1}{\delta}))$
\end{proof}

\subsection{Conclusion and Future Work}
We gave a new approach to constructing sampling-based sketches which can be quickly applied to rank-one tensors with at most $3$ modes.  To demonstrate the utility of this approach we showed how to speed up sketching for $\ell_0$ estimation, and for constructing $\ell_1$ embeddings.  

A number of intriguing questions remain.  For example, are there other constructions of $p$-samples that apply to higher mode tensors?  Additionally, our sketches require $O(n)$ time to construct each entry of the sketch.  While this can be faster than the $O(n^2)$ or $O(n^3)$ required to expand a rank-one tensor, it seems reasonable to hope for sketches that takes closer to $O(1)$ time per entry. This is the case for CountSketch based constructions that work by sampling and then hashing entries of the sample into buckets. Unfortunately in our setting it is not clear how to efficiently compose our $p$-samples with a CountSketch.

Finally, it would interesting to extend our ideas to other problems where sampling-based sketches have been applied.

\section{Experiments}

We evaluate the correctness of the $\ell_0$ samplers. Our $\ell_0$-sampler is theoretically guaranteed to output a uniformly random entry of the support, up to some constant factor.  That is, the probability we output a fixed entry of the support of a tensor $X$ is $[\frac{c_1}{|\supp(X)|}, \frac{c_2}{|\supp(X)|}]$ for some absolute constants $c_1$ and $c_2.$  In order to keep our analysis simple, the constants $c_1$ and $c_2$ gotten from unwinding our proof are more extreme than necessary.  We remedy this by empirically showing that our sampler is in fact much closer to uniform.

We choose our $\ell_0$ sampler to have $10$ buckets at each sampling level.  For an $N\times N \times N$ tensor $X$ our sampling rates begin at $1/N^3$ and increase in powers of $5.$  While these parameters are somewhat differnt from what we chose theoretically, we find that they give good practical results.

In all experiments we work with $3$-mode tensors.
We consider two different support structures for our experiments.  The first support shape we consider, we call the ``disjoint rectangle" support.  In this model, the support consists of two boxes $B_1$ and $B_2$ of dimensions $(x_1,x_2,x_3)$ and $(y_1, y_2, y_3),$ such that

Recall that the first step of $p$-sample construction is to permute each mode separately, so up to permutation there is only one configuration of rectangles to consider.  Moreover, up to this permutation symmetry, all entries of $B_1$ are equivalent to one another, and all entries of $B_2$ are equivalent to one another.  Thus to test for uniformity, it suffices to check that the number of samples from $B_1$ is approximately $|B_1|/(|B_1| + |B_2|).$

The second support shape that we consider is a rectangle $B$ of dimensions $(x_1, x_2, x_3)$, along with an additional $x_1 x_2 x_3$ random entries (sampled from outside the rectangle) which we call $R$, the random component of the support.  Note that we choose the size of $R$ to be the same as the size of $B.$  This is mainly for convenience, so that our $\ell_0$ tester is successful if approximately half the samples lie in $R.$

In each of experiments we use the $p$-sample construction that we previously discussed to build an $\ell_0$ tester.  To make the experiments run faster, we assume access to a perfect singleton tester.  While such a singleton tester does not exist in practice, it can be approximated arbitrarily well using ideas that we previously discussed.  Here our aim is to understand how well our $p$-sample construction works as a proxy for a truly uniform sample.

All of our experiments suggest that our $\ell_0$ sampling procedure behaves very nearly perfectly on the tensors described above.  All experiments are accurate to within a few percent of what one would expect for uniform sampling.  See section~\ref{sec:experiment_data} in the appendix for tables showing our experimental data. 

The code for the experiments, along with additional implementations of our sketch are available at \url{https://github.com/wswartworth/tensorSampling}.




\printbibliography

\newpage
\appendix
\onecolumn

\section{Construction of an $\ell_1$ Embedding}

We design a sketch for $\ell_1$ embeddings similar to \cite{verbin2012rademacher}.  However we show how to design the sketch in such a way so as to allow fast application to rank-one tensors.



We describe the general structure of our sketch here.  We will later choose parameters in order to obtain the desired guarantees.
\paragraph{Description of sketch.}  Our sketch $\ellonesketch$ consists of a series of $\ell$ sampling levels numbered $0, \ldots, \ell-1$.  At each level $h$, there are two parameters: $T_h$ and $p_h.$  $T_h$ is the number of buckets at level $h$, and $p_h$ is the sampling probability at level $h.$  Each bucket in level $h$ is formed by taking a random sign combination of coordinates corresponding to a $(p_h/T_h)$-sample.  So if $S$ is a $(p_h/T_h)$-sample, a bucket in level $h$ takes the form
\[
\frac{1}{p_h}\sum_{i \in S} \sigma_i x_i,
\]
where $x\in \R^n$ is the vector that we sketch, and $\sigma \in\{+1,-1\}^n$.  We will choose $\sigma$ so that $\sigma_i$ is a Kronecker product of $q$ random sign vectors, each of which is $4$-wise independent. 
We use the notation $B_1^{(h)}, \ldots, B_{T_h}^{(h)}$ to represent the buckets at sampling level $h.$  We will allow $p_h$ to be larger than $1$, as long as $p_h/T_h \leq 1.$


\paragraph{$\ell_1$ sketch analysis} We first give the following technical lemma, which is given here to shorten the proof of the following lemma.
\begin{lemma}
\label{lem:chebyshev_chernoff_fusion}
Let $x\in\R^d$ have entries in $[0,\alpha].$  Let $X_1, \ldots, X_k \sim X$ be i.i.d. with $\Pr(X = \abs{x_j}) \geq p$ for all $j.$  Then for $k \geq c\log\frac{1}{\delta}\max(\frac{\alpha}{p\norm{x}{1}},1)$, we have $\frac{1}{k}(X_1 + \ldots + X_k) \geq \frac{1}{4} p \norm{x}{1}$ with probability at least $1-\delta.$
\end{lemma}

\begin{proof}
It suffices to consider the case where $\Pr(X = x_j) = p$ for all $j$ and is $0$ with the remaining probability, since this majorizes the random variable in the lemma.  Note that $\E(X_1 + \ldots + X_k) = pk\norm{x}{1}$.  We also have 
\[
\E((X_1 + \ldots + X_k)^2) = pk\norm{x}{2}^2 + k(k-1)p^2\norm{x}{1}^2,
\]
So
\[
\text{Var}(X_1 + \ldots + X_k) \leq pk\norm{x}{2}^2 \leq pk\alpha\norm{x}{1}.
\]

For $k\geq \frac{16\alpha}{p\norm{x}{1}}$ we have $\sqrt{pk\alpha\norm{x}{1}} \leq \frac{1}{4}pk\norm{x}{1}$, and so
by Chebyshev's inequality,
\[
\Pr(X_1 + \ldots + X_k \leq \frac{1}{2}pk\norm{x}{1})
\leq \frac{1}{4}.
\]

If instead we have $k\geq (24\log(1/\delta))\cdot \max(\frac{16\alpha}{p\norm{x}{1}},1)$, then group the sum into $24\log(1/\delta)$ terms (each group of size at least one), each of which is at least half their expectation with $3/4$ probability.  By a Chernoff bound, at least half the groups will be at least this large with $1-\delta$ probability, which means that the entire sum will be at least $1/4$ its expectation, as desired.
\end{proof}

Our next task is to lower bound the contribution of a given sampling level.

\begin{lemma}
\label{lem:no_contraction}
Suppose that the number of modes $q$ is at most $3$. Consider a sampling level of our sketch as described above, with sampling probability $p$ and $T$ buckets. Suppose that $\norm{x}{1} = 1$.  Fix $m$ and $\gamma$ and let $L$ be the set of coordinates of $x$ with magnitude in $(m, \gamma m]$.  Then for $(p/T)|L| \leq 1/4,$ $T \geq c\log(1/\delta)$ and $p \geq c\log(1/\delta) \frac{\gamma m}{\norm{x_L}{1}}$, we have
\[
\abs{B_1} + \ldots + \abs{B_T} \geq c\norm{x_L}{1}
\]
with probability at least $1-\delta.$
\end{lemma}
\begin{proof}
For convenience, set $p' = p/T.$ Consider a fixed bucket $B_k$ formed from $p'$-sample $S_k$ and random signs $\sigma_k$.  Also consider fixed $j\in L.$  The probability that $j \in S_k$ is at least $p'/2$ by definition of the sampler.  

Recall that $\sigma_k$ is a product of $q\in\{2,3\}$ $4$-wise independent sign vectors, say $\tau_1, \tau_2, \tau_3$  We claim that conditioned on $j\in S_k$, we have $\abs{B_k} \geq \frac{c}{p}\abs{x_j}$ with constant probability.  To see this, consider first sketching by $\tau_1$ to compress along the first mode.  For the fiber $y$ of $x$ that contains $j$, $\abs{\inner{\tau_1}{y}}$ is at least $c\norm{y}{2} \geq c\norm{x_j}{2}$ with constant probability, since $\tau_1$ has $4$-wise independent signs \footnote{This follows for example from the classic analysis of the AMS sketch\cite{alon1996space}. Taking the average of $O(1)$ such sign measurements gives a $1 \pm (1/3)$ multiplicative approximation to the $\ell_2$ norm of $y$ with $2/3$ probability say (by a Chebyshev bound). Thus, each measurement individually must be $c\norm{y}{2}$ with constant probability.}. Then iteratively apply the same argument to sketching by $\tau_2$ and $\tau_3$ along the remaining modes.


Using the assumption that $p'|L|\leq 1/4,$ we have that for $i\in L$, $\Pr(|S_k \cap L| = \{i\}) \geq p'/4$ by \ref{prop:basic_psample_facts}. Applying the previous Lemma~\ref{lem:chebyshev_chernoff_fusion} with $p'/8$ gives that 
\[
\frac{1}{T}(p\abs{B_1} + \ldots + p\abs{B_T}) \geq \frac{1}{4}(p'/8)\norm{x}{1}
\]
with probability at least $1-\delta$ when $T \geq c \log\frac{1}{\delta} \max(\frac{\gamma m}{(p'/8) \norm{x_L}{1}},1)$.  Recalling that $p' = p/T,$ this rearranges to the desired bound.







\end{proof}

\subsection{No Dilation}

The following proposition will help to bound the contribution of the small coordinates to a given sampling level.  This will be the key place where we take advantage of cancellation.  Unlike earlier analyses such as \cite{verbin2012rademacher}, we do not have independence, so instead we will show how to get away with a Markov bound (earlier analyses could have applied a similar argument).

\begin{proposition}
Suppose that $x\in \R^n$ with $\norm{x}{1}\leq 1$ and $\norm{x}{\infty}\leq \alpha.$  Let $\sigma_1, \ldots, \sigma_n \in \{0, -1, 1\}$ be such that $\E(\sigma_i \sigma_j) = 0$ for all $i\neq j$ and $\E(\abs{\sigma_i}) \leq q.$  Then
\[
\E(\abs{\sigma_1 y_1 + \ldots + \sigma_n y_n})
\leq \min(\sqrt{\alpha q},q).
\]
\end{proposition}
\begin{proof}
By Jensen's inequality,
\begin{align*}
\E(\abs{\sigma_1 y_1 + \ldots + \sigma_n y_n})^2
&\leq \E\abs{\sigma_1 y_1 + \ldots + \sigma_n y_n}^2\\
&= \E(\sigma_1^2 y_1^2 + \ldots + \sigma_n^2 y_n^2)\\
&\leq qy_1^2 + \ldots + qy_n^2\\
& = q \norm{y}{2}^2.
\end{align*}
Note that $\norm{y}{2}^2 \leq \norm{y}{1}\norm{y}{\infty} \leq \alpha$, so the first part of the bound follows.  For the second half, we have
\begin{align*}
\E(\abs{\sigma_1 y_1 + \ldots + \sigma_n y_n})
&\leq \E(\abs{\sigma_1}\abs{y_1} + \ldots +\abs{\sigma_n}\abs{y_n} )\\
&\leq q \norm{y}{1} \leq q.
\end{align*}

\end{proof}

\begin{lemma}
\label{lem:general_dilation_bound}
Consider applying our $\ell_1$-embedding sketch $\ellonesketch$ to a vector $x.$ Then with with failure probability at most, $p_1/\beta_1 + \ldots + p_{\ell}/\beta_{\ell} + 1/10$, we have the bound
\[
\norm{\ellonesketch x}{1} \leq 10\left(\sum_i \norm{x_{[\alpha_i, \beta_i]}}{1} + \sum_i \min(\sqrt{T_i \alpha_i/p_i}, 1)\right)
\]
\end{lemma}
\begin{proof}

Let $E_i$ be the event of sampling a coordinate larger than $\beta_i$ in layer $i$.  There are at most $1/\beta_i$ such coordinates, so the probability that $E_i$ occurs is at most $p/\beta_i$ by a union bound.

Now let $x_{\leq \alpha_i}$ be $x$ where all coordinates larger in magnitude than $\alpha_i$ are zeroed out.  The previous lemma bounds the expected contribution of $x_{\leq \alpha_i}$ to layer $i$ by $ \min(\sqrt{T\alpha_i/q}, 1)$ (after rescaling).

Now let $x_{[\alpha_i, \beta_i]}$ contain all coordinates of $x$ between $\alpha_i$ and $\beta_i$ with the remaining coordinates zeroed out.  The expected contribution of $x_{[\alpha_i, \beta_i]}$ to layer $i$ is $\norm{x_{[\alpha_i, \beta_i]}}{1}$ since the layers are unbiased.

Thus with failure probability at most, $p_1/\beta_1 + \ldots + p_{\ell}/\beta_{\ell} + 1/10$, we have the bound
\[
10\left(\sum_i \norm{x_{[\alpha_i, \beta_i]}}{1} + \sum_i \min(\sqrt{T_i \alpha_i/p_i}, 1)\right)
\]
on the total $\ell_1$ norm of the sketch.
\end{proof}

As a corollary, a standard net argument, such as that given in \cite{clarkson2014sketching} for example, extends our embedding result for a single vector to a subspace.
\begin{corollary}
Let $A \in \R^{n\times d}$ be a matrix.  There is an oblivious sketch $\ellonesketch$ with sketching dimension $m = O(d^2\log n + \log^4 n)$ such that
\[
c_1 \norm{Av}{1} \leq \norm{\ellonesketch Av}{1} \leq c_2\norm{Av}{1},
\]
where $0 < c_1 < c_2$ are absolute constants.  Moreover all entries of $\ellonesketch$ can be taken to be in $\{0,+1,-1\}$, and $\ellonesketch$ can be applied to a tensor in $(\R^k)^{\otimes 2}$ in $O(mk)$ time and can be applied to a tensor in $(\R^k)^{\otimes 3}$ in $O(mk\log^2 k)$ time.
\end{corollary}

\paragraph{Combining the bounds}
\begin{theorem}
There is an $O(1)$-distortion $\ell_1$ embedding sketch $\ellonesketch$ $\R^n \rightarrow \R^m$ with sketching dimension $m = O(\log^4 n + \log^2(1/\delta)\log n)$ that satisfies $\norm{\ellonesketch x}{1} \leq c \norm{x}{1}$ with constant probability, and satisfies $\norm{\ellonesketch x}{1} \geq c \norm{x}{1}$ with probability at least $1-\delta.$

Moreover our sketch can be applied to rank one tensors in $\R^{k}\otimes \R^{k}$ in $O(m k)$ time, and to rank one tensors in $\R^{k}\otimes \R^{k} \otimes \R^k$ in $O(m k\log^2 k)$ time.
\end{theorem}
\begin{proof}
    For $h\geq 0$, we set $p_h = q^h$ for a fixed $q$ and set $T_h = T$ for a fixed value of $T$, both to be determined.  For notational convenience we will also define a layer $-1$ of the sketch with $p_{-1} = q^{-1}$ and $T_{-1}=T.$  While $q^{-1}$ is larger than $1,$ $T$ will be chosen so that $q^{-1}/T < 1$.

    \paragraph{Contraction Bound.} Divide the coordinates of $x$ into levels $L_i$ for $i\geq 0$, with $L_i$ consisting of the coordinates of $x$ in $(q^{i+1},q^i]$ and all other coordinates replaced with $0.$  We say that a level $L_i$ is heavy if $\norm{x_{L_i}}{1} \geq 1/(10\ell).$

    Suppose that $L_i$ is heavy.  We will choose parameters so that level $i-1$ of the sketch, preserves a constant fraction of $L_i$'s mass.  Setting $p=q^{i-1}$ in Lemma~\ref{lem:no_contraction}, the conditions we need are $q^{i-1}|L_i|/T \leq 1/4$, $T \geq c\log(1/\delta)$ and $q^{i-1} \geq c\log(1/\delta)\frac{q^i}{\norm{x_{L_i}}{1}}.$  Under the assumption that $L_i$ is heavy, we have $\norm{x_{L_i}}{1} \geq \frac{1}{10\ell}$.  Also $\abs{L_i} \leq 1/q^{i+1}$ since $\norm{x}{1}=1.$  Therefore it suffices to choose parameters so that the following relations hold:
    \begin{enumerate}
    \item $\frac{1}{q} \geq c (\log\frac{1}{\delta})\ell$
    \item $T \geq \frac{c}{q^2}$
    \end{enumerate}
    (Note that the condition $T \geq c \log\frac{1}{\delta}$ is redundant in light of the two listed.)

    \paragraph{Dilation Bound.} Next we bound the dilation of the sketch on $x$ with $\norm{x}{1} = 1.$  In the dilation bound above, we choose $\alpha_i = p^{i+3}$ and $\beta_i = p^{i-1}.$

    Then the first sum in the dilation bound Lemma~\ref{lem:general_dilation_bound}, is bounded by $4$ since each coordinate of $x$ appears in the sum at most $4$ times.  For the second sum, we have
    \[
    \sum_{i=1}^{\ell} \min(\sqrt{T_i\alpha_i/p_i}, 1) \leq
    \sum_{i=1}^{\ell} \sqrt{Tq^{i+3}/q^i}
    = \sum_{i=1}^{\ell} \sqrt{T q^3}.
    \]
    We will choose parameters so that $Tq^3 < 1/\ell^2$, so that this sum will be bounded by $1.$   Then
    \[
    \sum_{i=0}^{\ell}\min(\sqrt{T_i\alpha_i/p_i}, 1) \leq 1 + 1 = 2.
    \]

    We therefore get an $O(1)$ dilation bound with failure probability at most
    $\frac{1}{10} + \sum_{i=1}^{\ell}\frac{q^i}{q^{i-1}} \leq q \ell.$ Thus we must choose $q \leq c/\ell.$

\paragraph{Choosing Parameters.} Taking stock of our constraints on parameters, we have
\begin{enumerate}
\item $\frac{1}{q} \geq c(\log\frac{1}{\delta})\ell$
\item $T \geq \frac{c}{q^2}$
\item $Tq^3 < 1/\ell^2$
\item $q\leq \frac{c}{\ell}$.
\end{enumerate}
These can be satisfied by setting $\frac{1}{q} = c \max(\log\frac{1}{\delta} \ell, \ell^2)$ and $T = c/q^2.$  Since $\ell \leq \log n$ this yields an overall bound of $T = O(\log^4 n + \log^2(1/\delta)\log^2 n).$  Since there are $T$ buckets in each of $\ell$ levels, this gives the desired bound on the sketching dimension.

Finally the time bounds follow by using our construction of $p$-samples that admit fast summation: we first apply the $\sigma_i$'s given in our construction, which takes $O(n)$ time.  We then compute the sum of the resulting tensor entries over a $p$-sample that admits $\tilde{O}(n)$ summation. This takes $O(n)$  time in the two mode case, and $O(n\log^2 n)$ time in the three mode case.
\end{proof}

\section{Proof of Lemma 2.9}
Let $i\in \supp(X)$ By \ref{prop:basic_psample_facts}, we have 
\[
\Pr(S\cap \supp(X) = \{i\}) \geq \frac{1}{2}p(1 - 2p|\supp(X)|)
\geq \frac{1}{4}p
.
\] 
From the definition of a $p$-sample we have 
\[
\Pr(S\cap \supp(X) = \{i\}) \leq \Pr(i \in S) \leq p.
\]
Also from \ref{prop:basic_psample_facts}, 
\begin{align*}
\Pr(|S\cap \supp(X)|=1) &\geq \frac{1}{2}|\supp(X)|p(1 - 2p|\supp(X)|)\\
&\geq \frac14 \abs{\supp(X)}p.
\end{align*}
By taking a union bound over $i\in \supp(X),$ we have 
\[
\Pr(|S\cap \supp(X)|=1) \leq |\supp(X)|p.
\]
It follows from this that
\begin{align*}
\Pr(i \in S \big| |S\cap \supp(X)| = 1)
&= \frac{\Pr(S\cap \supp(X) = \{i\})}{\Pr(|S\cap \supp(X)| = 1)}\\
&\geq \frac{(1/4)p}{|\supp(X) p|} = \frac{1}{4|\supp(X)|}
\end{align*}
and 
\begin{align*}
\Pr(i \in S \big| |S\cap \supp(X)| = 1)
&\leq \frac{p}{(1/4)|\supp(X)|p}\\
&= \frac{4}{|\supp(X)|}.
\end{align*}


\section{Experiment Data}
\label{sec:experiment_data}

All experiments are carried out on a tensor of shape $40\times 40 \times 40$ for $1000$ trials.  For our $\ell_0$ sampler, we take our levels to have sampling probabilities increasing in powers of $5.$  We additionally use $10$ buckets per level.  The following tables give the raw data for our experiments. 

\begin{table}
\centering
\begin{tabular}{|c|c|c|c|}
\hline
Rectangle dimensions & Fraction in first rectangle & Expected Fraction & Failures \\
\hline
((1, 1, 20), (1, 1, 1)) & 0.9650 & 0.9524  & 0 \\
((1, 10, 20), (1, 1, 1)) & 0.9970 & 0.9950  & 0 \\
((1, 20, 20), (1, 1, 1)) & 0.9980 & 0.9975  & 1 \\
((20, 20, 20), (1, 1, 1)) & 1.0000 & 0.9999  & 1 \\
((1, 1, 20), (10, 10, 10)) & 0.0200 & 0.0196  & 0 \\
((1, 10, 20), (10, 10, 10)) & 0.1610 & 0.1667  & 0 \\
((1, 20, 20), (10, 10, 10)) & 0.2970 & 0.2857  & 0 \\
((20, 20, 20), (10, 10, 10)) & 0.8959 & 0.8889  & 1 \\
((1, 1, 20), (20, 20, 20)) & 0.0010 & 0.0025  & 3 \\
((1, 10, 20), (20, 20, 20)) & 0.0160 & 0.0244  & 2 \\
((1, 20, 20), (20, 20, 20)) & 0.0500 & 0.0476  & 0 \\
((20, 20, 20), (20, 20, 20)) & 0.5130 & 0.5000  & 0 \\
\hline
\end{tabular}
\caption{In each experiment, the support consists of two disjoint boxes.  The dimensions of the two boxes are given in the leftmost column.  The table shows the fraction of points that our $\ell_0$ sampler chooses in the first rectangle, as well as the fraction that we expect for a perfect $\ell_0$ sampler.  The rightmost column records the number of times our $\ell_0$ sampler failed.}
\label{table:disjoint_rects}
\end{table}

\begin{longtable}[h!]{|c|c|c|}
\hline
Rectangle shape & Fraction in rectangle & Failures \\
\hline
(1, 1, 1) & 0.47 & 5 \\
(1, 1, 3) & 0.51 & 0 \\
(1, 1, 9) & 0.52 & 2 \\
(1, 1, 27) & 0.52 & 0 \\
(1, 3, 1) & 0.52 & 1 \\
(1, 3, 3) & 0.51 & 4 \\
(1, 3, 9) & 0.50 & 0 \\
(1, 9, 1) & 0.51 & 3 \\
(1, 9, 3) & 0.51 & 0 \\
(1, 9, 9) & 0.50 & 0 \\
(3, 1, 1) & 0.49 & 1 \\
(3, 1, 3) & 0.50 & 1 \\
(3, 1, 9) & 0.52 & 0 \\
(3, 3, 1) & 0.49 & 0 \\
(3, 3, 3) & 0.53 & 0 \\
(3, 3, 9) & 0.51 & 0 \\
(3, 9, 1) & 0.51 & 0 \\
(3, 9, 3) & 0.49 & 0 \\
(3, 9, 9) & 0.51 & 0 \\
(9, 1, 1) & 0.53 & 1 \\
(9, 1, 3) & 0.50 & 0 \\
(9, 1, 9) & 0.51 & 0 \\
(9, 3, 1) & 0.51 & 0 \\
(9, 3, 3) & 0.49 & 0 \\
(9, 3, 9) & 0.53 & 0 \\
(9, 9, 1) & 0.54 & 0 \\
(9, 9, 3) & 0.50 & 1 \\
(9, 9, 9) & 0.52 & 1 \\
(1, 3, 27) & 0.49 & 1 \\
(1, 9, 27) & 0.50 & 0 \\
(1, 27, 1) & 0.51 & 0 \\
(1, 27, 3) & 0.50 & 0 \\
(1, 27, 9) & 0.50 & 0 \\
(1, 27, 27) & 0.48 & 0 \\
(3, 1, 27) & 0.52 & 1 \\
(3, 3, 27) & 0.51 & 1 \\
(3, 9, 27) & 0.53 & 0 \\
(3, 27, 1) & 0.51 & 0 \\
(3, 27, 3) & 0.50 & 0 \\
(3, 27, 9) & 0.53 & 2 \\
(3, 27, 27) & 0.47 & 3 \\
(9, 1, 27) & 0.48 & 1 \\
(9, 3, 27) & 0.53 & 0 \\
(9, 9, 27) & 0.51 & 1 \\
(9, 27, 1) & 0.50 & 1 \\
(9, 27, 3) & 0.47 & 1 \\
(9, 27, 9) & 0.50 & 1 \\
(9, 27, 27) & 0.49 & 0 \\
(27, 1, 1) & 0.53 & 0 \\
(27, 1, 3) & 0.54 & 0 \\
(27, 1, 9) & 0.52 & 2 \\
(27, 1, 27) & 0.50 & 2 \\
(27, 3, 1) & 0.52 & 0 \\
(27, 3, 3) & 0.49 & 2 \\
(27, 3, 9) & 0.50 & 0 \\
(27, 3, 27) & 0.50 & 0 \\
(27, 9, 1) & 0.50 & 1 \\
(27, 9, 3) & 0.50 & 0 \\
(27, 9, 9) & 0.51 & 0 \\
(27, 9, 27) & 0.53 & 0 \\
(27, 27, 1) & 0.47 & 0 \\
(27, 27, 3) & 0.51 & 1 \\
(27, 27, 9) & 0.48 & 0 \\
(27, 27, 27) & 0.48 & 0 \\
\hline
\caption{Each row of the table corresponds to an experiment run on a support shape consisting of a box of dimension specified by the left column of the table, as well as an equal number of additional uniformly sampled entries.  We run $1000$ trials with our $\ell_0$ sampler, and record the fraction of samples from successful runs of the sampler that are in the rectangle.  The rightmost column records the number of times the sampler failed.}
\label{tabel:rect_and_rand}
\end{longtable}

\end{document}